\documentclass[11pt,a4paper]{article}

\usepackage{ifpdf}

\makeatletter
\let\@font@warningori\@font@warning
\newcommand\shutup{\def\@font@warning##1{}}
\newcommand\youcanspeak{\let\@font@warning\@font@warningori}
\makeatother

\usepackage{amsmath,amsbsy,bm,latexsym,enumerate,tabularx}

\shutup
\usepackage{fourier}
\youcanspeak
\usepackage[scaled=0.875]{helvet}

\usepackage[latin1]{inputenc}
\usepackage[german,dutch,english]{babel}

\usepackage{amsthm}

\usepackage{graphicx,epsfig,color}

\usepackage{tikz}
\usetikzlibrary{calc}
\usetikzlibrary{intersections}

\usepackage{anysize}

\usepackage{subfigure}

\newcommand{\smat}[1]{ \big(\begin{smallmatrix} #1 \end{smallmatrix}\big)}

\providecommand{\setN}{\mathbb{N}}
\providecommand{\setZ}{\mathbb{Z}}
\providecommand{\setQ}{\mathbb{Q}}
\providecommand{\setR}{\mathbb{R}}

\newcommand {\cP}         {\mathcal{P}}

\newcommand {\Z}          {\mathbb{Z}}
\newcommand {\Q}          {\mathbb{Q}}
\newcommand {\R}          {\mathbb{R}}
\newcommand {\N}          {\mathbb{N}}

\newcommand{\vol}{\mathrm{vol}}

\renewcommand{\epsilon}{\varepsilon}

\renewcommand{\leq}{\leqslant}
\renewcommand{\geq}{\geqslant}

\theoremstyle{plain}
\newtheorem{theorem}{Theorem}
\newtheorem{lemma}[theorem]{Lemma}

\theoremstyle{definition}

\title{Minimizing the number of lattice points in a translated polygon}
\author{Friedrich Eisenbrand\footnote{Supported by the Alexander von
    Humboldt Foundation (AvH) and the German Research Foundation
    (DFG)} \\
   Technische Universität Berlin\\
  \texttt{eisenbrand@tu-berlin.de}
   \and
    Nicolai Hähnle\\
   Technische Universität Berlin\\
  \texttt{haehnle@math.tu-berlin.de}}

\date{\today}

\begin{document}

\maketitle

\begin{abstract}
  \noindent
  The parametric lattice-point counting problem is as follows: Given
  an integer matrix $A \in \Z^{m × n}$, compute an explicit formula
  parameterized by $b \in \setR^m$ that determines the number of integer points in
  the polyhedron $\{x \in \R^n : Ax \leq b\}$.  In the last decade, this
  counting problem has received considerable attention in the
  literature.  Several variants of Barvinok's algorithm have been
  shown to solve this problem in polynomial time if the number  $n$ of
  columns of $A$ is fixed.  Central to our
  investigation is the following question:
  \begin{quote}
    Can one also efficiently  determine a parameter $b$
    such that the number of integer points in $\{x \in \R^n : Ax \leq b\}$
    is minimized?
  \end{quote}
  Here, the parameter $b$ can be chosen from a given polyhedron
  $Q\subseteq\setR^m$.
  Our main result is a proof that finding such a minimizing parameter
  is $NP$-hard, even in dimension $2$ and even if the parametrization
  reflects a translation of a $2$-dimensional convex polygon.
  This
  result is established via a relationship of this problem to
  arithmetic progressions and simultaneous Diophantine approximation.

  On the positive side we show that in dimension $2$ there exists a
  polynomial time algorithm for each fixed $k$ that either determines
  a minimizing translation or asserts that any translation contains at
  most $1 + 1/k$ times the minimal number of lattice points.
\end{abstract}

\section{Introduction}

As many combinatorial optimization problems can be formulated as an
integer program, also their corresponding counting problems can be
formulated as the problem of counting integer points in a
polytope. Although both problems are hard in general, they can be
efficiently solved if the number of variables is fixed.  This was
shown by Lenstra~\cite{Lenstra83} in the case of integer programming and
by Barvinok~\cite{MR1304623} for the integer point counting
problem, see also \cite{Barvinok02,MR2271992}.

A \emph{parametric polyhedron} is a set of the form $P_b = \{ x \in
\setR^n \colon Ax\leq b\}$ for some matrix $A \in \setR^{m×n}$ and $b \in \setR^m$. The
right-hand-side $b$ is the \emph{parameter} of $P_b$. %It is restricted
%to be in some polyhedron $Q \subseteq\setR^m$.
Barvinok and
Pommersheim~\cite{BarvinokPommersheim99} extended Barvinok's integer
point counting algorithm to the parametric case. They describe an
algorithm that runs in polynomial time if the dimension is fixed and
that computes a \emph{quasipolynomial} %\marginpar{was ist das eigentlich?}
whose value at $b$ equals $|P_b \cap \setZ^n|$.  Since then, several
authors described alternative and more efficient algorithms to compute
this quasipolynomial~\cite{MR2383436,VerdoolaegeWoods2007}.
 Effective
implementations of Barvinok's algorithm have been provided by De~Loera
et al.~\cite{MR2094541} and by Köppe and Verdoolaege~\cite{MR2383436}.
Applications of parametric integer counting can, for example, be found
in compiler optimization~\cite{MR2312001}.
Other very interesting, though not polynomial-time approaches to the
parametric integer counting problem can for example be found in
\cite{MR1797548,MR1931190,MR2161199}.

In this paper we consider the problem of finding a parameter $b \in Q$
such that $|P_b \cap Z^n|$ is minimized. More precisely, we consider the
following decision problem.
\begin{quote}
  \textsc{Integer Point Minimization}

  \begin{tabular}{ll}
    Given:& $A \in \Q^{m × n}$, a rational polyhedron $Q \subseteq \R^{m }$, $k \in \N$ \\
    Decide:& $\exists b \in Q  :~ | \{ x \in \Z^n ~:~ Ax \leq b \} | \leq k$ \\
  \end{tabular}
\end{quote}
We remark that if $k = 0$ and $n$ is fixed, then this problem can be
solved in polynomial time with a technique of Kannan~\cite{Kannan92},
see also~\cite{EisenbrandShmonin2008}.

\subsubsection*{Contributions of this paper}

Our first result is a proof that integer point
minimization is $NP$-complete, even
if $n=2$, i.e. the parametric polyhedron resides in the Euclidean
plane and then even if the parametric polyhedra are the translations
of  some convex polygon along the $x$-axis.  In other words, we show
that the following problem is $NP$-complete.

\begin{quote}
\textsc{Polygon Translation}
  \begin{tabbing}
  Given:  \quad  \=  $A \in \setQ^{m × 2}$ and $b \in \setQ^m$ defining
  a convex polygon $P = \{ x \in \setR^2 \colon Ax\leq b\}\subseteq\setR^2$
   and $k \in \setN$ \\
  Decide: \>
   $\exists \lambda$,  $0 \leq\lambda \leq1$
   such that
   $
     P + \smat{-\lambda \\0} = \{ x +\smat{-\lambda \\0} \colon x \in P\}
     $
   contains at most $k$ integer points.
 \end{tabbing}
\end{quote}
Clearly, this is an instance of the parametric integer counting
problem with $Q\subseteq\setR^m$ being the 1-dimensional polytope $Q = \{ x \in \setR^m \colon x = b
- \lambda a_1, \, 0 \leq \lambda \leq1\}$, with  $a_1$ being the first column of
$A$.

\bigskip

\noindent
Second, we  show that there is a polynomial-time approximation
scheme for the optimization version of \textsc{Polygon Translation}.
More precisely,
there exists an algorithm that runs in polynomial time for any fixed
integer $k$ and
either determines a minimizing translation or asserts that any
translation contains at most $1 + 1/k$ times the minimal number of
lattice points.

This result combines techniques from the geometry of
numbers with classical techniques from discrepancy theory. The
\emph{discrepancy} of a polygon is the absolute value of the difference between
the number of integer points in the polygon and its area.
There is a rich literature bounding this
discrepancy, see, e.g.~\cite{MR1420620}, starting with
\emph{Gauss' circle problem}. Gauss~\cite{MR616130} investigated
the discrepancy of a disk of radius $R$ around $0$.
He showed that this discrepancy is bounded by $O(R)$,
which implies that the number of integer points is $\Theta(R^2)$ because the area of the disk is $\pi R^2$.
Discrepancy bounds also exist for polygons~\cite{MR1420620},
but they involve the \emph{length} of the boundary.
Instead, we bound the discrepancy in terms of the \emph{lattice width}
of the input polygon:
the number of integer points in a polygon of high lattice width is very close to its area.
On the other hand,
we adapt a technique of Kannan~\cite{MR1105115} to solve instances with thin polygons exactly.

% We will show that this more general problem is $NP$-hard even when $n = 2$ and $Q \cong \R$.
% The corresponding maximization problem can be seen to be $NP$-hard as well,
% using some simple adaptations to the proofs given in this paper.

% Integer programming is a natural $NP$-complete problem in the sense
% that many $NP$-complete problems can be easily formulated as an
% integer program.  Parametric integer programming plays a similar role
% for the class of $\Pi_2^p$-complete problems, since the class $\Pi_2^p$
% can be roughly thought of as the class of $\forall\exists$-decision
% problems~\cite{MR0438810}:
% \begin{quote}
%   \textsc{Parametric Integer Programming}

%   \begin{tabular}{ll}
%   Given:& $A \in \Q^{m × n}$, a rational polyhedron $Q \subseteq \R^{m + p}$ \\
%   Decide: &
%   $\forall b \in Q / \Z^p ~ \exists x \in \Z^n :~ Ax \leq b$ \\
%   \end{tabular}
% \end{quote}
% The set $Q / \Z^p$ is called the integer projection of $Q$ and is
% defined as follows:
% \[ Q / \Z^p := \left\{ b \in \R^m ~:~ \exists z \in \Z^p \text{ such that }
%   (b,z) \in Q \right\} \] Kannan~\cite{MR1105115} gave a polynomial
% time algorithm for this problem for the case where $n$, $p$, and the
% affine dimension of $Q$ are fixed.  Eisenbrand and
% Shmonin~\cite{MR2464645} removed the dependence on the affine
% dimension of $Q$.  These results were used to study integrality gaps,
% and to obtain algorithms for restricted cases of the Frobenius
% problem~\cite{MR1179254}.

% \subsection{Collection of related papers}

% \begin{itemize}
%   \item Hardness of EDF-schedulability~\cite{ER10}.
% \end{itemize}

\section{Polygon translation is NP-complete}
\label{sec:hardness}

% In this section, we show that Integer Point
% Minimization is $NP$-hard even in the following  restricted case.
% \begin{quote}
%   Given  a matrix $A \in \setQ^{m × 2}$, a vector $b \in \setQ^m$ defining
%   a convex polygon $P = \{ x \in \setR^2 \colon Ax\leq b\}$ in the
%   euclidean plane and $k \in \setN$,  decide whether there exists a
%   parameter $\lambda$,  $0 \leq\lambda \leq1$
%    such that
%    \begin{displaymath}
%      P + \smat{-\lambda \\0} = \{ x +\smat{-\lambda \\0} \colon x \in P\}
%    \end{displaymath}
%    contains at most $k$ integer points.
% \end{quote}
% %
% In other words, we are looking for a translated copy of $P$, where the
% translation is along the horizontal axis, such that the number of
% of integer points is in this translation is bounded by $k$. %This is an
% %instance of integer point minimization with $A$ and $k$ as above and

In this section, we provide our main hardness result. First, we
discuss the relation of \emph{arithmetic progressions} with the
polygon translation problem.

\subsection{Arithmetic progressions and their pulse functions}

The \emph{arithmetic progression} defined by the triple $(a,k,d)$ is the set
$A = \{ a, a + d, a + 2d, \dots, a + kd \}$.
We say that a function of the form
\[
  p(x) = \begin{cases}
    0 & \mbox{if } |x - y| < \varepsilon \mbox{ for some } y \in A \\
    1 & \mbox{else}
  \end{cases}
\]
where $\varepsilon > 0$ is a \emph{pulse function}.
The next lemma establishes a relation of pulse functions with the
polygon translation problem.
Figure~\ref{fig:quadrilateral-for-pulse-function} illustrates the
construction with an example.
\begin{figure}
  \begin{center}
    \begin{tikzpicture}
      \draw[very thick,red] (0,0) -- (0.28,0);
      \draw[very thick,red] (2,1) -- (2.53,1);
      \draw[very thick,red] (4,2) -- (4.78,2);

      \foreach \y in {0,1,2} \foreach \x in {0,1,...,10} \draw[help
      lines,fill] (\x,\y)  circle (1pt);% +(-2pt,-2pt) -- +(2pt,2pt) +(-2pt,2pt) --
%      +(2pt,-2pt);

      \draw[very thick] (0.28,0) node[below] {$(\ell_1,y_1)$} -- (4.78,2) node[above] {$(\ell_2,y_2)$}
        -- (7.62,2) node[above] {$(r_2,y_2)$} -- (9.12,0) node[below]
        {$(r_1,y_1)$} --cycle;
        \draw[very thick,green] (9,0) -- (9.12,0);
        \draw[very thick,green] (8,1) -- (8.37,1);
        \draw[very thick,green] (7,2) -- (7.62,2);

    \end{tikzpicture}
  \end{center}
  \caption{A quadrilateral for a pulse function $p(x)$ defined by the
    arithmetic progression $\{0.2, 0.45, 0.7\}$
  with $\varepsilon = 0.08$. The red line segments have length $0.28$, $0.53$
  and $0.78$ (bottom up) and the green line segments have length
  $0.12, 0.37$ and $0.62$ (bottom up). As the polygon is translated
  by $x$ to the left, the number of integer points inside the polygon
  is $17 + p(x)$. }
  \label{fig:quadrilateral-for-pulse-function}
\end{figure}
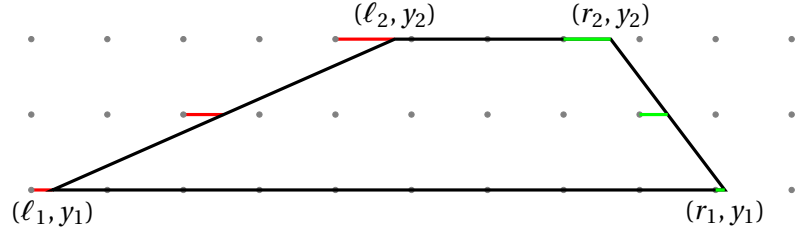
%
%The next lemma is straightforward. It relates pulse functions to more
%general quadrilaterals, see
%figure~\ref{fig:quadrilateral-for-pulse-function}.
\begin{lemma}
  \label{lemma:quadrilateral-for-pulse-function}
  Let $p : \R \to \{ 0, 1 \}$ be a pulse function for the arithmetic progression
  $A = \{ a, a + d, \ldots, a + kd \}$ with parameter $0 < \varepsilon \leq \frac{d}{2}$
  and discontinuities only in $(0,1)$. That is, $a - \varepsilon > 0$ and $a + kd + \varepsilon < 1$.

  Let $y_1, y_2 \in \Z$ such that $y_2 - y_1 = k$.
  Consider a convex quadrilateral $P$
  with vertices $(\ell_1, y_1)$, $(r_1, y_1)$, $(\ell_2, y_2)$, and $(r_2, y_2)$.
  If we have
  \begin{enumerate}
    \item $\ell_1 < r_1$ and $\ell_2 < r_2$,
    \item $\{ \ell_1 \} = a + \varepsilon$ and $\{ \ell_2 \} = a + kd + \varepsilon$,
    \item $\{ r_1 \} = a - \varepsilon$ and $\{ r_2 \} = a + kd - \varepsilon$, and
    \item $k ~|~ (\lfloor \ell_2 \rfloor - \lfloor \ell_1 \rfloor)$ and
          $k ~|~ (\lfloor r_2 \rfloor - \lfloor r_1 \rfloor)$,
  \end{enumerate}
  then there exists an $M \in \N$ so that
  \begin{displaymath}
       | (t \begin{pmatrix} -1 \\ 0 \end{pmatrix} + P) \cap \Z^2 | = M +
       p(\{t\}) \text{ for } t \in [0,1].
  \end{displaymath}

\end{lemma}
 \begin{proof}
   Consider the horizontal slice $L_i = P \cap \{ (x,y) ~|~ y = y_1 + i
   \}$ for $0 \leq i \leq k$.  Using conditions (2) to (4) of the Lemma,
   one verifies that $L_i$ is a line segment $[\alpha_i,\beta_i] × \{ y_1 + i
   \}$ with $\{ \alpha_i \} = a + id + \varepsilon$ and $\{ \beta_i \} = a + id - \varepsilon$.
   In other words, as the segment sweeps left, an integer point leaves
   at all times $t$ with $\{t\} = \{ \beta_i \} = a + id - \varepsilon$, and an
   integer point enters at all times $t$ with $\{t\} = \{ \alpha_i \} = a
   + id + \varepsilon$.  Taking into account that $L_i$ is relatively closed,
   one has
   \[ | (t \begin{pmatrix} -1 \\ 0 \end{pmatrix} + L_i) \cap \Z^2 | =
   M_i +
      \begin{cases}
        0 & \mbox{if } a + id - \varepsilon < \{ t \} < a + id + \varepsilon \\
        1 & \mbox{otherwise}
      \end{cases}
   \]
   for some $M_i \in \N$.
   In fact, $M_i = \lfloor \beta_i \rfloor - \lceil \alpha_i \rceil$.

   Summing over all $L_i$, $0 \leq i \leq k$,
   and using the fact that the intervals $(a + id - \varepsilon, a + id + \varepsilon)$
   are pairwise disjoint,
   we obtain the claim of the Lemma.
 \end{proof}

 % Finally, we will reduce Arithmetic Progression Meeting
%  to the following problem about the number of integer points
%  in translates of a given polygon:
%  \begin{quote}
%    \textsc{Polygon Translation}

%    \begin{tabular}{ll}
%      Given:& convex polygon $P \subset \R^2$, $v \in \Z^2$, $M \in \N$
%      \\
%      Decide:& $\exists t \in \R :~ |(tv + P) \cap \Z^2| \leq M$
%    \end{tabular}
%  \end{quote}
%  This problem is then easily seen to be a special case of Integer Point Minimization:
%  given an instance of Polygon Translation,
%  we can find a system $Ax \leq b$ of inequalities describing $P$.
%  But then we have for every $t \in \R$:
%  \[
%    tv + P = \{ tv + x \in \R^2 ~:~ Ax \leq b \} = \{ x' \in \R^2 ~:~ Ax' \leq b - A(tv) \},
%  \]
%  which means that we can choose $Q = \{ b - A(tv) ~:~ t \in \R \} \cong \R$
%  as the set of parameters to reformulate the problem in terms of Integer Point Minimization.
%  This proves $NP$-hardness of Integer Point Minimization
%  even when $d = 2$ and $Q \cong \R$.
%  We only have to fill in the details of the reductions mentioned above,
%  which is done in Theorems~\ref{thm:hardness-apm} and~\ref{thm:hardness-polygon-translation}.

Next we consider a decision problem involving several pulse
functions.
\begin{quote}
  \textsc{Arithmetic Progression Meeting}

  \begin{tabular}{ll}
    Given:& pulse functions $p_1$, \dots, $p_n$
      (encoded as their parameters $a^{(j)}$, $k^{(j)}$, $d^{(j)}$, $\varepsilon^{(j)} \in \Q$)
    \\
    Decide:& $\exists x \in \R :~ p(x) = \sum_{j=1}^n p_j(x) = 0$
  \end{tabular}
\end{quote}
We delay the proof of the next theorem to
Section~\ref{sec:simult-dioph-appr}.
\begin{theorem}
  \label{thm:hardness-apm}
  Arithmetic Progression Meeting is $NP$-hard.
\end{theorem}

%minimization. % Once, this result is established, it is clear that the
% translation
% problem for \emph{non-convex polygons} in the plane is hard via the
% following argument. The non-convex union of parallelograms $P_{
%   a^{(j)}, k^{(j)}, d^{(j)}, \varepsilon^{(j)}} + \smat{2\cdot i \\0}$, $i\leq0
% \leq n-1$ is non-overlapping. The sum of pulse functions $p(x) =
% \sum_{j=1}^n p_j(x)$ has a common root, if and only if this arrangement
% of parallelograms can be translated vertically, such that the total
% number of integer points contained in this arrangement is
% $\sum_{j=1}^nk^{(j)}$. The question will then be, how to turn this
% arrangement into a convex polygon that serves the same purpose. The
% next lemma gives us the desired flexibility in construction
% quadrilaterals whose number of integer points, as one translates them
% to the left, is the sum of a constant and the pulse function of the
% arithmetic progression. The final construction will stack the
% corresponding quadrilaterals on top of each other such that the result
% is a convex polygon, see Figure~\ref{fig:polygon-translation-stack}.
%
%
We are now in the position to prove that \textsc{Polygon Translation}
is $NP$-complete.
In our reduction to arithmetic progression meeting, we restrict ourselves to pulse functions whose
discontinuities lie in the open interval $(0,1)$. % ,
% and aim for a construction where all discontinuities $t$ of the counting function correspond
% to a discontinuity $x = \{ t \}$ of a pulse function, and vice
% versa.
We can always reduce to this special case using an affine
transformation of the pulse functions,
so that this restriction is without loss of generality.
%We begin with the case of a single pulse function,
%which is illustrated in Figure~\ref{fig:quadrilateral-for-pulse-function}.

\begin{theorem}
  \label{thm:hardness-polygon-translation}
  Polygon Translation is $NP$-hard.
\end{theorem}
\begin{proof}
  Let $p_1, \ldots, p_n : \R \to \{ 0,1 \}$ be an instance of Arithmetic Progression Meeting.
  We will assume without loss of generality that all discontinuities of the $p_j$
  lie in the open interval $(0,1)$.
  Let $p = \sum_{j=1}^n p_j$.
  The goal is to construct a convex polygon $P$ such that
  \[ | (t \begin{pmatrix} -1 \\ 0 \end{pmatrix} + P) \cap \Z^2 | = M + p(\{ t \}) \mbox{ for all } t \in \R \]
  for some $M \in \N$.
  Then $P$, $v = \begin{pmatrix} -1 \\ 0 \end{pmatrix}$, and $M$ form an instance
  of Polygon Translation which is a Yes-instance if and only if
  $p_1, \ldots, p_n$ is a Yes-instance for Arithmetic Progression Meeting.

  The idea for the construction of $P$ is straightforward.
  Lemma~\ref{lemma:quadrilateral-for-pulse-function} gives us a tool
  for constructing quadrilaterals $P_1, \ldots, P_n$ corresponding to the pulse functions $p_1, \ldots, p_n$,
  which we then stack vertically to form the polygon $P$ with $2n + 2$ edges.
  This is illustrated in Figure~\ref{fig:polygon-translation-stack}.
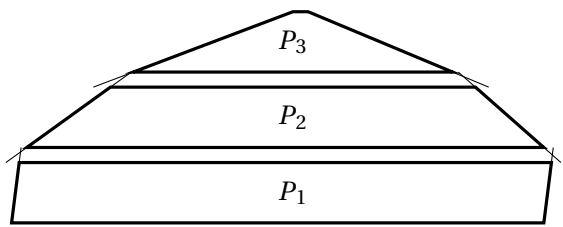
\begin{figure}
  \begin{center}
    \begin{tikzpicture}
      \draw[very thick] (0,0) -- (0.1,0.8) -- (7.1,0.8) -- (7.0,0) --cycle;
      \draw (0,0) -- (0.125,1.0);
      \draw (7,0) -- (7.124,1.0);
      \node at (3.7,0.4) {$P_1$};

      \draw[very thick] (0.2,1.0) -- (1.3,1.8) -- (6.1,1.8) -- (7.0,1.0) --cycle;
      \draw (-0.075,0.8) -- (1.575,2.0);
      \draw (5.875,2.0) -- (7.225,0.8);
      \node at (3.7,1.4) {$P_2$};

      \draw[very thick] (1.6,2.0) -- (3.7,2.8) -- (3.9,2.8) -- (5.8,2.0) -- cycle;
      \draw (1.075,1.8) -- (3.7,2.8);
      \draw (3.9,2.8) -- (6.275,1.8);
      \node at (3.7,2.4) {$P_3$};
    \end{tikzpicture}
  \end{center}
  \caption{The construction of $P$ in the proof of Theorem~\ref{thm:hardness-polygon-translation}.}
  \label{fig:polygon-translation-stack}
\end{figure}

  Formally, we define the polygon $P$ as the convex polygon constrained by $2n+2$ lines:
  the $2n$ lines defined by the left and right edges of the $P_j$,
  and the lines through the bottom edge of $P_1$ and the top edge of $P_n$.
  We will argue that, given a proper choice of coordinates for the $P_j$,
  the resulting polygon $P$ satisfies the following properties:
  \begin{enumerate}
    \item $P$ has $2n+2$ edges: the bottom edge of $P_1$, the top edge of $P_n$,
      and $n$ edges each on the left and right sides,
      which are extensions of the left and right edges of the $P_j$,
      respectively.
    \item $P$ has $2n+2$ vertices: the two bottom vertices of $P_1$,
      the two top vertices of $P_n$,
      and $2(n-1)$ vertices which are obtained as the intersection points of
      lines through the left or right edges of adjacent $P_j$.
    \item Any horizontal translate of $P$ contains exactly the same integer points as
      the union of the corresponding translates of the $P_j$.
  \end{enumerate}
  The last property is the result that we really need for the reduction.
  The first two properties merely guide us along during the proof.

  We will use the same notation as in Lemma~\ref{lemma:quadrilateral-for-pulse-function},
  but with superscripts indicating which polygon $P_j$ we are talking about.
  We choose $y_1^{(j+1)} = y_2^{(j)} + 1$ for all $j = 1 \ldots n-1$,
  and $y_2^{(j)} = y_1^{(j)} + k^{(j)}$ for all $j = 1 \ldots n$.
  We then choose
  \begin{align*}
    \lfloor \ell_2^{(j)} \rfloor &= \lfloor \ell_1^{(j)} \rfloor + k^{(j)} \cdot (3j) \\
    \lfloor \ell_1^{(j+1)} \rfloor &= \lfloor \ell_2^{(j)} \rfloor + 3j + 2
  \end{align*}
  The fractional part of the $\ell_i^{(j)}$
  is chosen to satisfy the conditions of Lemma~\ref{lemma:quadrilateral-for-pulse-function}.
  Observe that this fixes the $y_i^{(j)}$ and $\ell_i^{(j)}$ up to an integer translation of $P$.

  Let us describe how our choice of the $\lfloor \ell_i^{(j)} \rfloor$
  establishes the first two properties listed above on the left side of $P$.
  Observe that the slopes of the left edges of the $P_j$
  are strictly decreasing: the slope of the left edge of $P_{j+1}$ is always strictly less
  than the slope of the left edge of $P_j$.

  Furthermore, we claim that the lines through the left edges of $P_j$ and $P_{j+1}$
  intersect in a point $(x,y)$ with $y_2^{(j)} < y < y_1^{(j+1)}$;
  that is, they intersect ``between'' $P_j$ and $P_{j+1}$.
  To see this, consider the point $(x',y')$ where the line through the left edge of $P_j$ intersects the
  horizontal line $y = y_1^{(j+1)}$.
  By our choice of the $\lfloor \ell_i^{(j)} \rfloor$,
  we know that
\[ x' < \ell_2^{(j)} + 3j + 1 < \lfloor \ell_2^{(j)} \rfloor + 3j + 2 = \lfloor \ell_1^{(j+1)} \rfloor. \]
  On the other hand, if $(x'',y'')$ is the intersection of the line through the left edge of $P_{(j+1)}$
  with the horizontal line $y' = y_2^{(j)}$,
  then
\[ x'' < \ell_1^{(j+1)} - 3(j+1) < \lfloor \ell_1^{(j+1)} \rfloor - 3j - 2 = \lfloor \ell_2^{(j)} \rfloor. \]
  This establishes the claim. The situation is illustrated in Figure~\ref{fig:left-edge-intersection}.
\begin{figure}
  \begin{center}
    \begin{tikzpicture}
      \foreach \p in {(1,0),(2,0),(3,0),(10,1),(11,1),(12,1)}
          \draw[help lines] \p   +(-2pt,-2pt) -- +(2pt,2pt) +(-2pt,2pt) -- +(2pt,-2pt);
      \draw[dotted] (5.5,0.5) -- (7.5,0.5);

      \draw[dotted] (0.5,0) -- (12.5,0) node[right] {$y = y_2^{(j)}$};
      \draw[dotted] (0.5,1) -- (12.5,1) node[right] {$y = y_1^{(j+1)}$};

      \draw[thick] (2.8,0) -- (5.3,0.3);
      \fill (2.8,0) circle (2pt) node[below] {$\ell_2^{(j)}$};
      \draw[help lines] (2,-0.3) -- (2,0);
      \node[below] at (2,-0.3) {$\lfloor \ell_2^{(j)} \rfloor$};
      \draw[thick] (10.3,1) -- (7.8,0.7);
      \fill (10.3,1) circle (2pt) node[above] {$x'$};

      \draw[thick] (1.1,0) -- (4.6,0.3);
      \draw[thick] (11.9,1) -- (8.4,0.7);
      \fill (11.9,1) circle (2pt) node[above] {$\ell_1^{(j+1)}$};
      \draw[help lines] (11,1.5) -- (11,1);
      \node[above] at (11,1.5) {$\lfloor \ell_1^{(j+1)} \rfloor$};
      \fill (1.1,0) circle (2pt) node[below] {$x''$};
    \end{tikzpicture}
  \end{center}
  \caption{Intersection of the lines through the left edges of $P_j$ and $P_{j+1}$
    in the proof of Theorem~\ref{thm:hardness-polygon-translation}.}
  \label{fig:left-edge-intersection}
\end{figure}
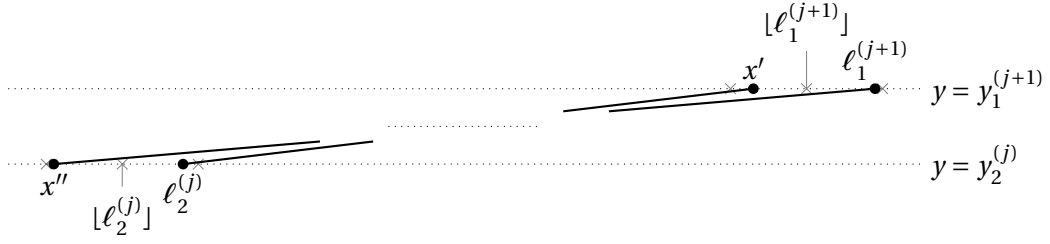

  An analogous choice of coordinates is used for the right edges of the $P_j$,
  to obtain negative slopes with strictly decreasing absolute value,
  and well-placed intersection points of the lines through right edges of adjacent $P_j$.

  Together, these properties imply that the boundary of $P$
  is indeed an extension of the boundaries of the $P_j$.
  In particular, every horizontal slice of $P$ with integer vertical coordinates
  coincides with a slice of one of the $P_j$, and vice versa.
  This establishes also the third property listed above.
  Hence
  \[
    | t \begin{pmatrix} -1 \\ 0 \end{pmatrix} + P | =
     \sum_{j=1}^n | t \begin{pmatrix} -1 \\ 0 \end{pmatrix} + P_j | =
     \sum_{j=1}^n (M^{(j)} + p_j(\{ t \})) =
     \underbrace{\left( \sum_{j=1}^n M^{(j)} \right)}_{=: M} + p(\{ t \})
  \]
  for all $t \in \R$, which completes the reduction.
\end{proof}

\subsection{Simultaneous Diophantine approximation and arithmetic
  progression meeting}
\label{sec:simult-dioph-appr}

Dirichlet's theorem is is a classical result in number theory
about the approximation of a vector of real numbers using rational numbers of equal denominator,
see e.g.~\cite[chapter 11]{MR2445243}.
It states that given $\alpha_1$, \dots, $\alpha_n \in \R$ and $Q \in \N$,
there exist $q \in \{ 1, \ldots, Q \}$ and $p_1$, \dots, $p_n \in \Z$ such that
\[ | \alpha_j - \frac{p_j}{q} | \leq \frac{1}{qQ^{1/n}} \qquad \mbox{ for all } j = 1\dots n \]
This result is best possible up to constants for the general case.
However, we may ask whether a better approximation is possible for a specific set of numbers.
This motivates the following decision problem,
which was shown to be $NP$-hard by Lagarias~\cite{Lagarias85}.
\begin{quote}
  \textsc{Simultaneous Diophantine Approximation}

  \begin{tabular}{ll}
    Given:& $\alpha_1,\ldots,\alpha_n\in\Q$, $Q\in\N$, $\varepsilon>0$ \\
    Decide:& $\exists q\in\{1,\ldots,Q\}$
    such that $|q\alpha_j - \lceil q\alpha_j \rfloor | \leq \varepsilon$ for all $j = 1\dots n$
  \end{tabular}
\end{quote}
We use $\lceil \cdot \rfloor$ to denote the nearest integer (breaking ties by rounding down).
For computational purposes, the $\alpha_j$ must be rational numbers,
though their denominators will typically be much larger than $Q$.

\begin{proof}[Proof of Theorem~\ref{thm:hardness-apm}]
  Let $\alpha_1, \ldots, \alpha_n \in \Q$, $Q \in \N$, $\varepsilon > 0$
  be an instance of Simultaneous Diophantine Approximation.
  We assume without loss of generality that $\alpha_j \in (0,1)$.

  We will define an instance of Arithmetic Progression Meeting with
  pulse functions $p_0$, $p_1$, \dots, $p_n$ as follows.
  We scale numerators and denominators so that the
  denominators of the $\alpha_j$ and $\varepsilon$ are all
  equal and we denote their common denominator by  $D$. % be the least common
% multiple of the denominator of $\varepsilon$ and the nominators of the
% $\alpha_j$.
  For every $j = 1\dots n$,
  let
\[
  p_j(x) = \begin{cases}
             0 & |x - \frac{i}{\alpha_j}| < \frac{\varepsilon}{\alpha_j} + \frac{1}{2D}
                 \mbox{ for some } i \in \{ 0, 1, \ldots, \lceil Q \alpha_j \rfloor \} \\
             1 & \mbox{else}
           \end{cases}
\]
  The intuition behind this definition is that we would like to have $p_j(x) = 0$ if and only if
  $| x \alpha_j - \lceil x\alpha_j\rfloor | \leq \varepsilon$ as in the definition of simultaneous Diophantine approximation.
  The correction term $\frac{1}{2D}$ is needed due to the strict inequality required in pulse functions.
  Furthermore, we define
\[
  p_0(x) = \begin{cases}
             0 & |x - i| < \frac{1}{2D}
                 \mbox{ for some } i \in \{ 1, 2, \ldots, Q \} \\
             1 & \mbox{else}
           \end{cases}
\]
  It remains to be shown that the original instance of Simultaneous Diophantine Approximation is a Yes-instance
  if and only if $p = \sum_{j=0}^n p_j$ has a root.

  Suppose $q \in \{ 1, \ldots, Q \}$ satisfies $| q\alpha_j - \lceil q\alpha_j \rfloor | \leq \varepsilon$
  for all $j = 1 \ldots n$. Then dividing by $\alpha_j$
  yields
  \[ | q - \frac{ \lceil q\alpha_j \rfloor }{ \alpha_j } | \leq \frac{\varepsilon}{\alpha_j} \]
  and hence one obtains $p(q) = 0$.

  Conversely, suppose that $p(q) = 0$.
  Define $\hat q := \lceil q \rfloor$.
  As $p_0(q) = 0$, we have $\hat q \in \{ 1, \ldots, Q \}$.
  Furthermore, $|q - \hat q| < \frac{1}{2D}$. Let
  $i_j\in\{0,\ldots,\lceil Q\alpha_j\rfloor\}$ with $|q - i_j / \alpha_j| < \varepsilon/\alpha_j +
  1/(2D)$. Then
  \begin{eqnarray*}
    |q - i_j / \alpha_j| & = & |\hat q + q - \hat q - i_j / \alpha_j| \\
                & \geq & |\hat q  - i_j / \alpha_j| - |q - \hat q| \\
                & \geq & |\hat q  - i_j / \alpha_j| - 1/(2D) \\
  \end{eqnarray*}
From this, it follows that
\begin{displaymath}
  |\hat q  - i_j / \alpha_j| < \varepsilon / \alpha_j + 1/D
\end{displaymath}
and since $\alpha_j \in (0,1)$ also that
\begin{displaymath}
  |\hat q \alpha_j   - i_j | < \varepsilon  + 1/D.
\end{displaymath}
Since the denominator of $\alpha_j$ and $\varepsilon$ is $D$ one has
\begin{displaymath}
  |\hat q \alpha_j   - i_j | \leq \varepsilon
\end{displaymath}
which shows that $\hat{q}$ is a solution to simultaneous Diophantine
approximation.

%   We claim that for every $j = 1 \ldots n$
%   there exists an $i \in \Z$ such that
%   $|\hat q - \frac{i}{\alpha_j}| \leq \frac{\varepsilon}{\alpha_j}$,
%   i.e. $\hat q$ is the denominator of a good Diophantine approximation.
%   By contradiction, suppose that for a fixed $j$ we have
%   \[
%     |\hat q - \frac{i}{\alpha_j}| > \frac{\varepsilon}{\alpha_j} \mbox{ for all } i \in \Z.
%   \]
%   Then since $\hat q$ and $i$ are integral, and by our choice of $D$, we have that
%   \[
%     |\hat q - \frac{i}{\alpha_j}| \geq \frac{\varepsilon}{\alpha_j} + \frac{1}{D} \mbox{ for all } i \in \Z
%   \]
%   But taking into account that $|q - \hat q| < \frac{1}{2D}$, this contradicts $p_j(q) = 0$.
%   Hence the claim is established, and this completes the proof of correctness for our
%   reduction from Simultaneous Diophantine Approximation to Arithmetic Progression Meeting.
\end{proof}

\section{A polynomial time approximation scheme}

In this section, we show the following theorem,
which implies a polynomial time approximation scheme for the
\textsc{Polygon Translation} problem.
\begin{theorem}
  \label{thm:ptas}
  For every $k \in \N$,
  there is a polynomial-time algorithm which,
  given a polygon $P$ and a direction $v \in \Z^2$ as input,
  either computes a translate $tv + P$ containing a minimal number of integer points
  or asserts that every translate of $P$ is a $(1 + 1/k)$-approximation to the optimal solution.
\end{theorem}
The intuition behind this result is that when $P$ is small,
we can use integer programming in fixed dimension and adapt a technique of Kannan~\cite{MR1105115}
to find an optimum.
On the other hand, if $P$ is large and contains many lattice points,
then only a small fraction of them is close to the boundary,
and hence the \emph{discrepancy} relative to the average number of lattice points
that one expects based on the area of $P$ is small.

This idea goes back to Gauss' circle problem.
Gauss showed that as $r$ grows, the number of integer points $L(r)$ in a disk of radius $r$
is asymptotically equal to its area, $\pi r^2$.
In fact, he gave a bound on the error term $|L(r) - \pi r^2|$ that is linear in $r$.
The heart of the argument lies in counting unit squares intersecting the disk
and showing that only $O(r)$ of them lie near the boundary.
Observe that if one transforms the setting of Gauss' circle problem
by a linear map,
the disk becomes an ellipse $E$ and $\Z^2$ becomes a general lattice $\Lambda$.
Instead of unit squares we now count fundamental parallelepipeds of $\Lambda$;
the trick is to use the right parallelepiped.

The dual of a lattice $\Lambda$ is
$\Lambda^\star = \{ y \in \R^2 ~:~ \forall x\in\Lambda :~ y^Tx \in \Z \}$.
The \emph{width} $w_y(K)$ of a convex body $K$
along a dual lattice vector $y \in \Lambda^\star \setminus \{ 0 \}$
is defined as
\[ w_y(K) := \max_{x\in K} y^T x - \min_{x \in K} y^Tx. \]
The \emph{lattice width} $w(K)$ of $K$ is the minimum over all choices of $y \in \Lambda^\star \setminus \{ 0 \}$.
The lattice width and the corresponding dual lattice vector
can be computed efficiently in fixed dimension.
Note that in the linear transformation of Gauss' circle problem,
the diameter of the disk becomes the lattice width of the ellipse $E$.
In fact, the following theorem implies the discrepancy bound of Gauss.

\begin{theorem}
  \label{thm:wide-polygons}
  For every lattice $\Lambda \subset \R^2$ and convex body $K \subseteq \R^2$ with lattice width at least $k \geq 1$,
  one has $|N - \frac{\vol(K)}{\det(\Lambda)}| \leq \frac{3}{2k} \frac{\vol(K)}{\det(\Lambda)}$,
  where $N = |K \cap \Lambda|$.
\end{theorem}
\begin{proof}
  Using a linear transformation,
  we can assume that the Löwner-John ellipsoid~\cite{MR1491097} of $K$ is a unit disk centered at some $z \in \R^2$:
  \[ B(z,1) \subseteq K \subseteq B(z,2) \]
  Let $B = (b_1,b_2) $ be a reduced basis of the dual lattice $\Lambda^\star$.
  In other words, $\|b_1\| \leq \|b_2\|$ and its Gram-Schmidt orthogonalization satisfies
  \begin{equation}
    \label{eq:gram-schmidt-conditions}
    b_2 = b_2^\star + \mu b_1, \, b_2^\star \bot b_1, \, |\mu| \leq \frac{1}{2}.
  \end{equation}
  It is well known that $b_1$ is a shortest non-zero lattive vector.
  Thus it is the lattice width direction of disks.
  In particular, $4\|b_1\| = w(B(z,2)) \geq w(K) \geq k$.
  Let
  \[
    \cP = \{ x \in \R^2 ~:~ 0 \leq b_1^T x < 1 \text{ and } 0 \leq b_2^T x < 1 \}
  \]
  be the fundamental parallelepiped of $\Lambda$ associated to $B$ and let $\overline{\cP}$ be its closure.
  We relate the lattice points in $K$ to the area of $K$ by centering one copy of $\overline{\cP}$
  at each lattice point, see Fig.~\ref{fig:parallelepiped-tiling}:
  \[
    \cP(K) = (K \cap \Lambda) ± \frac{1}{2} \overline{\cP} = \bigcup_{x \in K \cap \Lambda} x ± \frac{1}{2} \overline{\cP}.
  \]
  Let $R = \partial K ± \frac{1}{2} \overline{\cP}$.
  We have
  \[ K \setminus R \subseteq \cP(K) \subseteq K \cup R. \]
  Let $\delta > 0$ be the radius of $\cP$, i.e. $\cP$ is contained in a disk of radius $\delta$.
  Since $K$ is a convex body, we can estimate
  \[
    \vol(R \setminus K) \leq \vol((K + B(0,\delta)) \setminus K) \leq |\partial K| \delta + \pi\delta^2,
  \]
  where $|\partial K|$ is the length of the boundary.
  By Lemma~\ref{lemma:diameter-of-parallelepiped} below,
  we have $\delta \leq \frac{1.2}{\|b_1\|} \leq \frac{1}{3k}$.
  Furthermore, $|\partial K| \leq |\partial B(z,2)| = 4\pi$ because $K$ is convex and contained in $B(z,2)$.
  \[
    \vol(R \setminus K) \leq \frac{4\pi}{3} \frac{1}{k} + \frac{\pi}{9k} \frac{1}{k} \leq \frac{3}{2k} \vol(K),
  \]
  where the last inequality follows from $\vol(K) \geq \vol(B(z,1)) = \pi$ and $k \geq 1$.
  It follows that
  \[ \vol(\cP(K)) \leq \vol(K) + \vol(R \setminus K) \leq \vol(K) + \frac{3}{2k} \vol(K). \]
  A lower bound follows from an analogous argument,
  and combining these inequalities with
  $\vol(\cP(K)) = N \cdot \vol(\cP) = N \cdot \det(\Lambda)$
  yields the statement of the Theorem.
\end{proof}

\begin{figure}
  \begin{center}
    \begin{tikzpicture}
      \def\drawP#1{
        \draw[fill=black!10] #1 +(0.1,0.633) -- +(-0.5,0.533) -- +(-0.1,-0.633) -- +(0.5,-0.533) --cycle;
      }
      \foreach \y / \x in {
        0/3,0/4,0/5,
        1/1,1/2,1/3,1/4,1/5,1/6,
        2/3,2/4,2/5,2/6
      } {
        \drawP{($(-2,-1.8) + \y *(-0.4,1.166) + \x *(0.6,0.1)$)}
      }

      %\draw (0,0) circle (1.5cm);
      % touch points: (0,-1.5), horiz; (1.3,0.75)
      \draw[thick] (-2,-1.5) -- (1,-1.5) -- (1.6, -0.5) -- (1.675,0.15) -- (0.55,2.05) --
        (-0.8,1.616) -- (-1.8,-0.116) --cycle;

%      \draw (-3.5,1.65) node[left] {$b_1^Tx = \beta$} -- (1.9,2.55);

      \foreach \y in {0,1,2,3}
        \foreach \x in {0,1,2,...,8}
          \draw[help lines] ($(-2,-1.8) + \y *(-0.4,1.166) + \x
          *(0.6,0.1)$)  circle (1pt);

%            +(-1.5pt,1.5pt) -- +(1.5pt,-1.5pt) +(-1.5pt,-1.5pt) -- +(1.5pt,1.5pt);

    \end{tikzpicture}
  \end{center}
  \caption{The set $\cP(K)$ from the proof of Theorem~\ref{thm:wide-polygons}.}
  \label{fig:parallelepiped-tiling}
\end{figure}
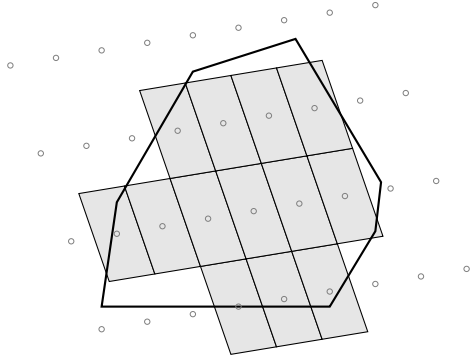

\begin{lemma}
  \label{lemma:diameter-of-parallelepiped}
  Let $B = (b_1,b_2)$ be a reduced basis of $\Lambda^\star$
  and $\cP = \{ x \in \R^2 ~:~ 0 \leq b_1^T x < 1 \text{ and } 0 \leq b_2^T x < 1 \}$
  the associated fundamental parallelepiped of $\Lambda$.
  Then the diameter $d$ of $\cP$ is bounded by $d \leq \frac{2.4}{\|b_1\|}$.
\end{lemma}
\begin{proof}
  Using the triangle inequality,
  we bound $d \leq \|x\| + \|y\|$, where $x$ and $y$ are vertices of $\cP$ adjacent to $0$.
  In particular, let $x$ be the vertex satisfying
  $b_1^Tx = 0$ and $b_2^Tx = 1$.
  Using notation from the proof of Theorem~\ref{thm:wide-polygons},
  we compute $(b_2^\star)^T x = b_2^Tx + \mu b_1^Tx = 1$.
  Since $x \bot b_1$ and hence $x$ is parallel to $b_2$, we get
  \[
    \|x\|^2 = \frac{1}{\| b_2^\star \|^2} \leq \frac{4}{3} \frac{1}{\|b_1\|^2}.
  \]
  The inequality follows from \eqref{eq:gram-schmidt-conditions} and $\|b_1\| \leq \|b_2\|$.
  Similarly, let $y$ be the vertex satisfying
  $b_1^Ty = 1$ and $b_2^Ty = 0$.
  We compute
  $(b_2^\star)^Ty = b_2^Ty - \mu b_1^Ty = -\mu$,
  and conclude using Pythagoras' theorem:
  \[
    \|y\|^2 \leq \frac{1}{\|b_1\|^2} + \frac{\mu^2}{\|b_2^\star\|^2}
      \leq \frac{1}{\|b_1\|^2} + \frac{4}{3} \frac{\mu^2}{\|b_1\|^2} \leq \frac{4}{3} \frac{1}{\|b_1\|^2}.
  \]
  The statement of the lemma follows from $2\sqrt{4/3} = 2.309\ldots$.
\end{proof}

In the second part of this section,
we will show how to find an optimal translate
when the lattice width of $P$ is at most a constant.
We extend a technique which was introduced by Kannan for parametric integer programming~\cite{MR1105115}.
Kannan determines the lattice width direction of the parametric polyhedron $P_b$
as a function of the parameter $b$,
and partitions the parameter space according to width direction and
according to how the respective lattice hyperplanes
interact with the boundary of the polyhedron.
In our case, the lattice width direction is the same for all parameters,
since we only translate the input polygon.
We partition the parameter space only based on interactions
of the boundary of $tv + P$ with the lattice hyperplanes orthogonal to the width direction of $P$.
Our main extension is that we encode counting the number of integer points on lattice slices
in an integer program, where Kannan's work only tested for feasibility.
Our approach is compatible with partitioning the parameter space based on the lattice width direction,
and hence the following Lemma can be extended to even more general
$2$-dimensional \textsc{Integer Point Minimization} problems,
provided that the lattice width is bounded by a constant
for all possible parameter values.

\begin{lemma}
  \label{lemma:thin-polygon-exact-solution}
  Given a dual lattice vector $y \in \Z^2 \setminus \{ 0 \}$ such that $w_y(P) \leq k$,
  the optimal translate of $P$ in direction $v \in \Z^2$
  can be computed in time $2^{O(k \log k)} b^{O(1)}$,
  where $b$ is the encoding length of $P$, $v$, and $y$.
\end{lemma}
\begin{proof}
  Using a unimodular transformation if necessary,
  we can assume without loss of generality that $y = e_1$.
  Let us sketch a simple algorithm to compute the number of integer points in a translate $tv + P$.
  Let us denote $\beta = \min_{x \in tv + P} e_1^T x$ the first coordinate of the leftmost point in the translate.
  Let
  \[
    S_i = (tv + P) \cap \{ x \in \R^2 ~:~ x_1 = \lceil \beta \rceil + i \}, i = 0 \dots k
  \]
  denote the integral vertical slices of $tv + P$.
  Note that some of the $S_i$ may be empty.
  For each slice, we can compute the lower end $a_i$ and upper end $b_i$
  and write $S_i = \{ \lceil \beta \rceil + i \} × [a_i, b_i]$.
  It follows that
  \[
    |(tv + P) \cap \Z^2| = \sum_{i=0}^k |S_i \cap \Z^2| = \sum_{i=0}^k \lfloor b_i \rfloor - \lceil a_i \rceil + 1
  \]
  We will argue that this algorithm can be encoded into a small number of integer programs
  that allow us to find the optimal $t$.
  We start by writing down the minimization of the number of integer points based on the $a_i$ and $b_i$.
  \begin{align*}
         \min & \sum_{i=0}^k y_i \\
     & y_i \geq B_i - A_i + 1 & (\forall i) \\
     & y_i \geq 0  & (\forall i) \\
     & b_i - 1 < B_i \leq b_i  & (\forall i) \\
     & a_i \leq A_i < a_i + 1  & (\forall i) \\
     & A_i, B_i \in \Z  & (\forall i)
  \end{align*}
  We can obtain $\gamma = \lceil \beta \rceil$ similarly:
  \begin{align*}
    & \beta = \beta_0 + tv_1 \\
    & \beta \leq \gamma < \beta + 1 \\
    & \gamma \in \Z,
  \end{align*}
  where we precompute $\beta_0$ as the value of $\beta$ for $t = 0$.
  The only remaining task is to encode the computation of the $a_i$ and $b_i$ given $t$ and $\gamma$.

  Suppose that we know which edge of $tv+P$ the point $(\gamma + i, a_i)$ lies on,
  and suppose that the corresponding edge of $P$ lies on the straight line defined by $c^Tx = d$.
  Then $a_i$ is defined uniquely by the equation
\[
  c^T(\gamma+i, a_i)^T - c^T tv = d,
\]
  hence we can express $a_i$ as a linear function in $t$ and $\gamma$
  and add a corresponding constraint to the integer program.

  Unfortunately, the point $(\gamma + i, a_i)$ does not always lie on the same edge.
  Let us separate the translation of the polygon into a horizontal and a vertical component,
  because a vertical translation does not affect the incidence between edges and vertical lines.
  As the polygon is translated horizontally,
  the point $(\gamma + i, a_i)$ moves onto a different edge of the polygon
  when a corresponding vertex of the polygon crosses the vertical line $x_1 = \gamma + i$.
  Over all points $(\gamma + i, a_i)$ and $(\gamma + i, b_i)$,
  such an event happens $n$ times -- once per vertex -- for one unit of horizontal movement.

  Hence we can separate $[0,1)$ into intervals $I_1, \ldots, I_n$
  with the property that the combinatorics of incidences between vertical lines and edges of the polygon
  are constant for all $\gamma - \beta$ within each interval.
  This allows us to solve one integer program for each of the intervals,
  each integer program with the added constrained that $\gamma - \beta \in I_j$ for some $j$,
  and appropriate constraints computing the $a_i$ and $b_i$ as outlined above.
  Together, these $n$ integer programs cover the entire space of possible values for $t$,
  and we simply take the best solution found among all of them.
  Each individual integer program has $O(k)$ variables, $2k + 1$ of which are integer variables,
  and can therefore be solved in time $2^{O(k \log k)} b^{O(1)}$
  using Kannan's algorithm for integer programming~\cite{Kannan87}.
\end{proof}

\begin{proof}[Proof of Theorem~\ref{thm:ptas}]
  We summarize the algorithm as follows:
  \begin{enumerate}
    \item Compute the lattice width and width direction $y \in \Z^2 \setminus \{0\}$ of $P$.
    \item If $w(P) \leq 4k$,
      compute an optimal translate using the algorithm of Lemma~\ref{lemma:thin-polygon-exact-solution}.
    \item Otherwise, assert that every translate is a $(1 + 1/k)$-approximate solution.
  \end{enumerate}
  The correctness of the last step follows from Theorem~\ref{thm:wide-polygons}:
  Let $A$ be the area of $P$ and let $OPT$ be the number of integer points in an optimal solution.
  Then for every $t \in \R$:
  \[
    |(tv + P) \cap \Z^2|
      \leq (1 + \frac{3}{8k})A \leq (1 + \frac{3}{8k}) \cdot \frac{1}{1-\frac{3}{8k}} \cdot OPT
      \leq (1 + \frac{1}{k}) OPT. \qedhere
  \]
\end{proof}

\small
\bibliographystyle{alpha}
\bibliography{papers,mybib,my_publications}

\end{document}